\documentclass[twocolumn,pre,amsmath,amssymb,showpacs,aps]{revtex4-1}

\usepackage{amsfonts} 
\usepackage{graphicx, subfigure}
\usepackage{color}
\usepackage{amsmath}
\usepackage{amssymb}
\usepackage{verbatim}
\usepackage{amsthm}
\usepackage{booktabs}
\usepackage{multirow}

\newtheorem{theorem}{Theorem}[section]
\newtheorem{lemma}[theorem]{Lemma}

\begin{document}



\preprint{APS/123-QED}

\title{The influence of the network topology on epidemic spreading}







\author{Daniel Smilkov}
\affiliation{Macedonian Academy for Sciences and Arts, Skopje, Macedonia \\
E-mail: dsmilkov@cs.manu.edu.mk}

\author{Ljupco Kocarev}
\affiliation{Macedonian Academy for Sciences and Arts, Skopje, Macedonia \\
 BioCircuits Institute, University of California, San Diego \\
 9500 Gilman Drive, La Jolla, CA 92093-0402 \\ E-mail: lkocarev@ucsd.edu}


\begin{abstract} 
The influence of the network's structure on the dynamics of spreading processes has been extensively studied in the last decade. Important results that partially answer this question show a weak connection between the macroscopic behavior of these processes and specific structural properties in the network, such as the largest eigenvalue of a topology related matrix.
However, little is known about the direct influence of the network topology on microscopic level, such as the influence of the (neighboring) network on the probability of a particular node's infection. To answer this question, we derive both an upper and a lower bound for the probability that a particular node is infective in a susceptible-infective-susceptible model for two cases of spreading processes: reactive and contact processes. The bounds are derived by considering the $n-$hop neighborhood of the node; the bounds are tighter as one uses a larger $n-$hop neighborhood to calculate them. Consequently, using local information for different neighborhood sizes, we assess the extent to which the topology influences the spreading process, thus providing also a strong macroscopic connection between the former and the latter. Our findings are complemented by numerical results for a real-world e-mail network. A very good estimate for the infection density $\rho$ is obtained using only 2-hop neighborhoods which account for 0.4\% of the entire network topology on average.

\end{abstract}


\maketitle

\section{Introduction} 
Complex network theory has opened the way for exploring many dynamical processes on large-scale systems consisting of individual components connected in a nontrivial topology. One of the most widely studied phenomena occurring on complex networks are spreading processes, with a prominent example attracting widespread attention being the spread of viruses in social or computer networks \cite{ref-3, ref-4, ref-5, ref-6, ref-7, ref-2008}.

There are several approaches being used in the analysis of epidemic spreading. One popular approach is the heterogeneous mean-field (HMF) prescription by coarse-graining nodes within degree classes and relaxing the problem by assuming that all nodes in a degree class have the same dynamical properties \cite{ref-4,ref-5, hmf1, hmf2}. However, it has been shown that HMF can result in different levels of accuracy \cite{hmf-accuracy}.
%
%
A more successful approach in determining the outcome of an infection was introduced by Chakrabarti et.al~\cite{ref-8} where the SIS epidemic model was analyzed by using a system of probability equations, which in fact, represents a deterministic non-linear dynamical system (NLDS).
%
This approach was also used in \cite{Gomez-2010}, where a family of SIS epidemic models is examined, parameterized by the number of stochastic contact trials per unit time, that range from contact processes (where the contagion expands at a certain rate from an infective vertex to one neighbor at a time) to reactive processes (in which an infective individual effectively contacts all its neighbors to expand the epidemics). Using a deterministic model, referred to as the Microscopic Markov-Chain approach (MMCA), which is virtually equivalent to NLDS, the whole phase diagram of the different infection models is constructed and their critical properties are determined. It is worth noting that using different number of stochastic contagion per unit time extends the usability of the model, since this number can surely vary for different real-world problems \cite{bounded-contact}. Recently, a mixed approach using both NLDS and HMF was proposed in \cite{npHMF} which lead to a nonperturbative formulation enhancing the predictive power of the classical HMF approach.
Heterogeneous environments have also been extensively studied. One such is an epidemic model with inhomogeneous infection probabilities on a graph with prescribed degree distribution \cite{ref-2011} where model's dynamics are derived for i.i.d. weights and for weights that are functions of the degrees.

With the help of these theoretical frameworks, the role of network topology in the spreading process has been repeatedly emphasized, yielding the result of a finite threshold for the spreading process in networks with exponentially bounded degree distributions, and a vanishing threshold in infinite uncorrelated networks with a power-law degree distribution. A recent addition to these findings is that for the SIS epidemic model, the vanishing threshold has nothing to do with the scale-free nature of the degree distribution, but is the result of the largest hub being a self-sustainable source for the infection \cite{psv-SIS} (see also \cite{reinfection}). However, the currently established connections are very rough with a topology-related threshold differentiating between two extreme outcomes of the model. With the threshold being satisfied,  there is still a large spectrum for different spreading parameters and the poorly understood role of the network's topology there motivated our work.
%

In this paper we adopt the approach proposed in \cite{ref-8, Gomez-2010} and study the deterministic epidemic model on graphs, in which the dynamics of individual nodes is described by a discrete-time Markov chain.
In the SIS model, a node can be in one of two states: susceptible (S) or infective (I). Infective nodes can infect other neighbouring nodes, and each node can be randomly cured with probability $\delta$ per unit time. At each time step, an infective node makes a number of trials per unit time to transmit the disease to its neighbours with probability $\beta$. We consider two specific cases: (i) the contact process, which involves a single stochastic contagion per infective node per unit time, and (ii) the reactive process, which involves as many stochastic contagions per unit time as neighbours a node has. 
The work in this paper extends that of \cite{ref-8,Gomez-2010}. We derive upper and lower bounds on the probability of a node to be infective, and determine how tight the bounds are around the probability that a node is infective. For both processes the bounds are derived using the $n-$hop neighbourhood of each node. The larger the considered neighbourhood -- the more topological information one uses to determine the bounds, hence the bounds are tighter. We use the difference between the upper and lower bound averaged over all nodes to determine the influence of the network topology on the spreading process and compute numerical results for a real-world e-mail network. For additional clarity, Figure \ref{fig:plotedgraphs} depicts the $1$-hop and $2$-hop neighborhood of a particular node in the Enron e-mail network with degree 10, together with the calculated bounds for its probability of infection derived using only its respective subgraph information.

The outline of the paper is the following. Section \ref{sec:2} gives the definition of the model and recovers known results. The contributions of this paper are contained in sections \ref{sec:3} and \ref{sec:4}. In section \ref{sec:3} the upper and lower bounds on the probability of being infective are derived for the reactive process, and numerical results for the e-mail network are presented. Section \ref{sec:4} gives the bounds for the contact process, along with the corresponding numerical results. Section \ref{sec:conc} concludes the paper and points out future research directions.

\begin{figure}[htbp]
  \begin{center}
   \subfigure[1-hop neighborhood]{
   \label{fig:plotedgraphs-1}
   \includegraphics[scale=0.3]{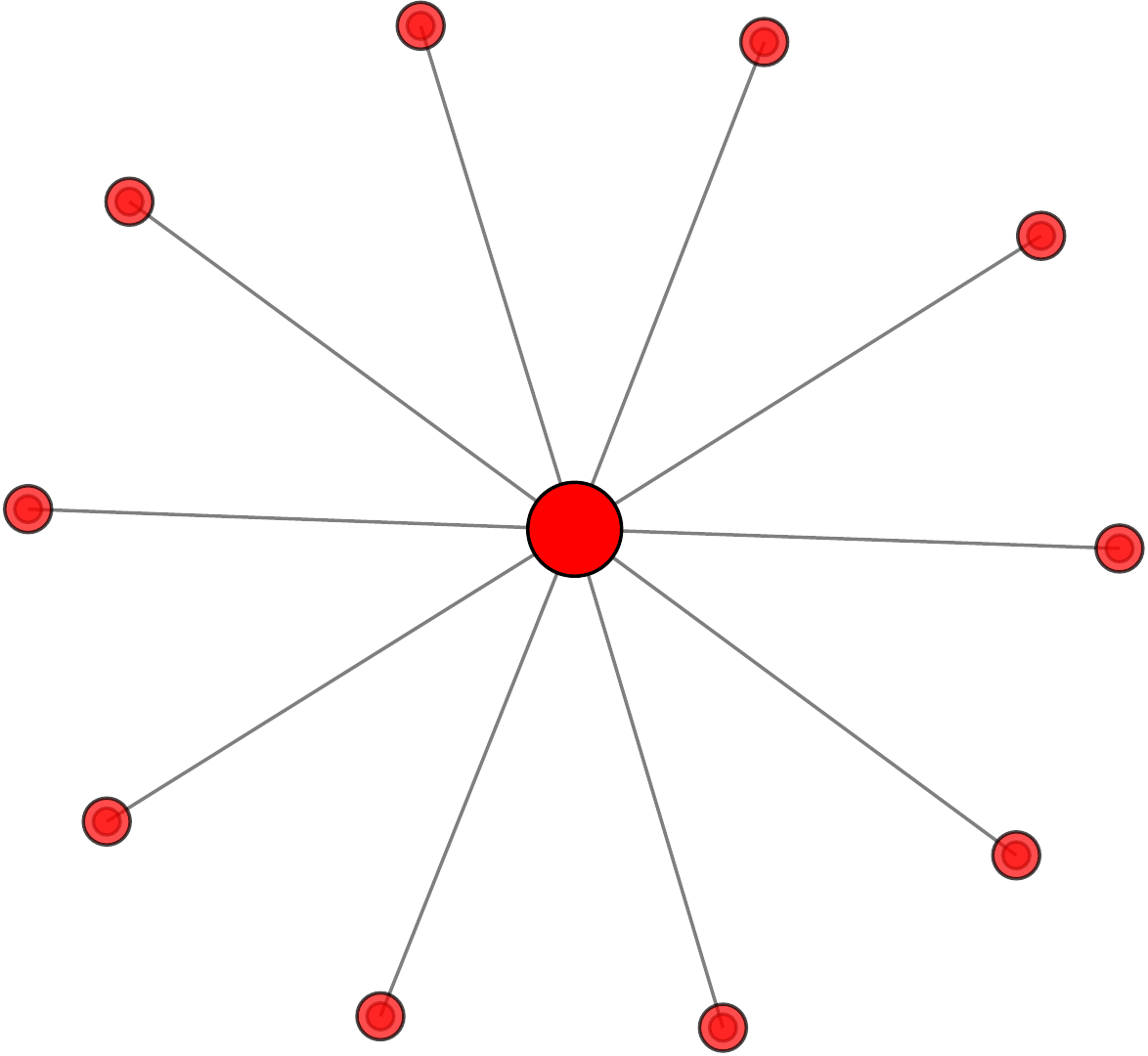}
  }
   \subfigure[2-hop neighborhood]{
   \label{fig:plotedgraphs-2}
   \includegraphics[scale=0.5]{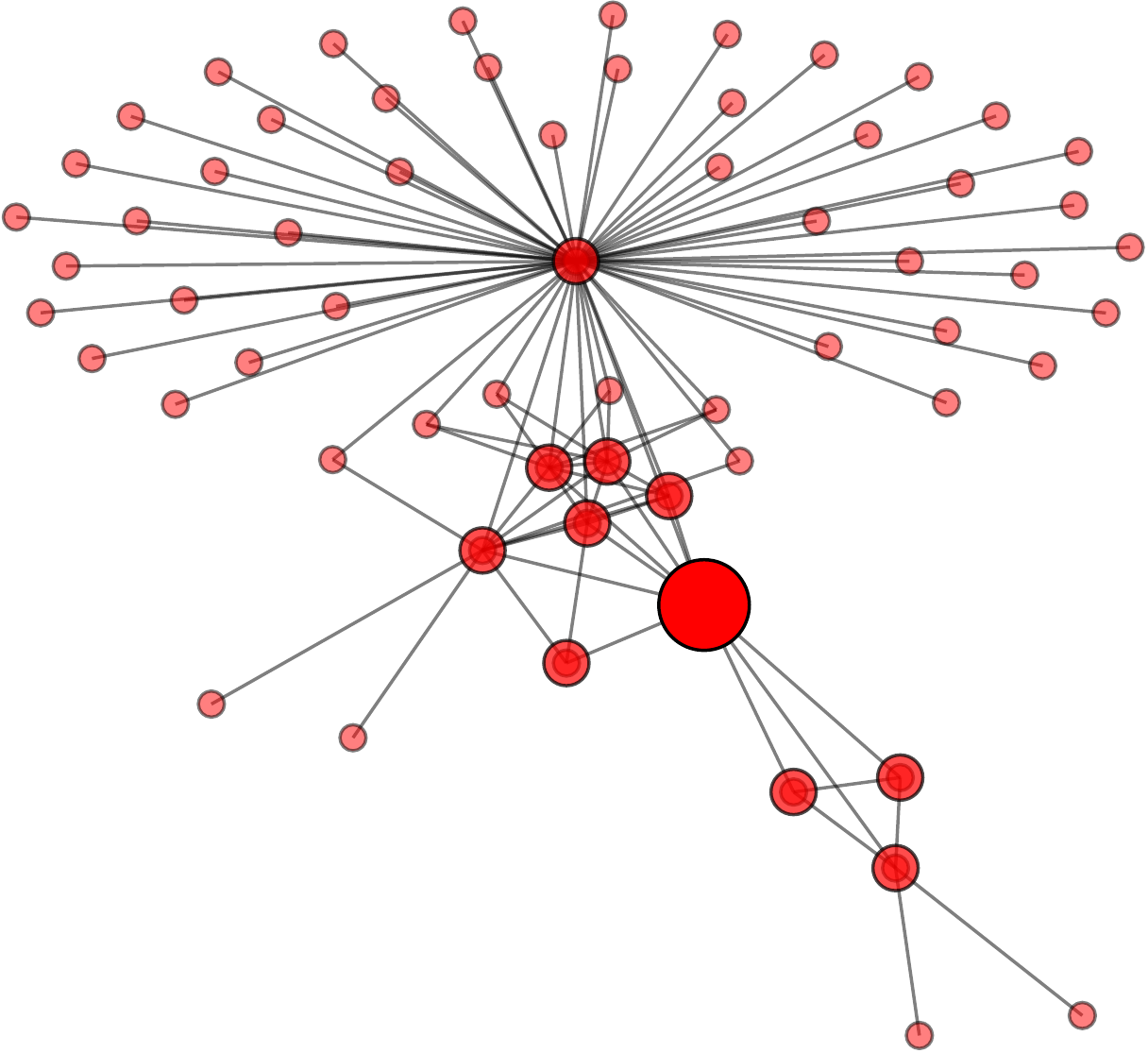}
  }
  \end{center}
  \caption{1-hop and 2-hop neighborhood for a node extracted from a real-world e-mail network with 33696 nodes. The node (largest in size) has 10 direct neighbors (medium sized). The probability of infection for the given node obtained after simulating a particular configuration of the SIS model was 0.373. \subref{fig:plotedgraphs-1} The 1-hop neighborhood represents a tree with the given node at its center and 10 peripheral nodes. The probability of infection for the node given the 1-hop neighborhood (node's degree) is calculated to be between 0.006 and 0.503. \subref{fig:plotedgraphs-2} The 2-hop neighborhood contains 62 nodes and 92 edges. Peripheral nodes are smallest in size and are two hops away from the central node. The probability of infection for the node given the 2-hop neighborhood topology is calculated to be between 0.297 and 0.416. Note that the difference between the upper and lower bound gets smaller as we use more topology information.}
  \label{fig:plotedgraphs}
\end{figure}

\section{Model definition and analysis}
\label{sec:2}

Consider a closed population of $N$ individuals, connected in a network structure which is represented by a simple, undirected,  unweighted, connected and unipartite graph $G = (V,E)$ with node set $V$ and edge set $E$. The adjacency matrix of the graph is given by $A = [ a_{ij} ]_{N \times N}$, where $a_{ij} = 1$ if node $i$ is connected to node $j$, and $a_{ij} = 0$ otherwise. 
Each node can be in one of two possible states: susceptible (S) and infective (I). Susceptible nodes are healthy and can contract the disease upon contact with infective nodes, which spread the disease. After the infectious period of the disease has ended, a node becomes susceptible to the disease once again. The initial set of infective nodes at time 0 is assumed to be non-empty, and all other nodes are assumed to be in state S at time 0.

The state of a node is represented by a status vector, an indicator vector containing a single 1 in the position corresponding to the present state, and 0 in the other $\mathbf{s}_i(t)=[s_i^S(t) \hspace{0.2cm} s_i^{I}(t)]^T,$
for all $i \in \{1, \ldots, N \}$. Let
$
\mathbf{p}_i(t)=[p_i^S(t) \hspace{0.2cm} p_i^{I}(t)]^T
$
be the probability mass function (PMF) of node $i$ at time $t$. 
The evolution of SIS is described by the following equations:
\begin{equation}
\begin{array}{lll}
p_i^S(t+1) &=& s_i^S (t) (1 - f_i(t)) + \delta s_i^I(t) \\
p_i^{I}(t+1) &=& s^S_{i}(t) f_i(t) + (1 - \delta) s_i^I(t)
\label{eq:stochastic-SIS}
\end{array}
\end{equation}
and
\begin{equation}
\mathbf{s}_i(t+1) = MultiRealize[\mathbf{p}_i(t+1)],
\label{eq:Multirealize}
\end{equation}
where $MultiRealize[\cdot]$ performs a random realization for the PMF given with $\mathbf{p}_i (t+1)$. In (\ref{eq:stochastic-SIS}) $0 \leq \delta \leq 1$ is the probability of curing and $0 \leq \beta \leq 1$ is the probability of disease transmission from an infective to a susceptible node.

We consider two cases of infection spreading: the contact process and the reactive process. The contact process \cite{contact-1,contact-2,contact-3} is a dynamical process that involves a single stochastic contagion per infective node per unit time, while in the reactive process \cite{reactive-1,reactive-2,reactive-3} there are as many stochastic contagions per unit time as there are neighbours to a node. The distinction between the two processes is reflected in the probability $f_i(t)$ that a susceptible node $i$ receives the infection from any combination of its infective neighbours. The probability $f_i(t)$ has the form:
\begin{equation}
f_i(t) = 1 - \prod_{j=1}^N (1 - \beta r_{ij} s_j^I(t)).
\label{eq:fi-stoc}
\end{equation}
where $r_{ij}$ is a contact probability. Without loss of generality, it
is instructive to think of these probabilities as the transition probabilities of random walkers on the network. The general case is represented by $\lambda_i$ random walkers leaving node $i$ at each time step:
$$
r_{ij} = 1 - \left( 1 - \frac{a_{ij}}{\sum_j a_{ij}}   \right)^{\lambda_i}
$$
The contact process corresponds to a model dynamics of one contact per unit time, $\lambda_i = 1$, $\forall i$, thus $r_{ij} = a_{ij}/\sum_j a_{ij}$. 
In the reactive process all neighbors are contacted, which corresponds, in this description, to set the limit $\lambda_i \to \infty$, $\forall i$,  resulting on $r_{ij} = a_{ij}$.

Though exact and realistic, the system of equations (\ref{eq:stochastic-SIS}) is not suitable for the analytical study of the system dynamics, since the new statuses are obtained as a result of a decision process, transforming a continuous variable into a discrete one. That is why, in the further text, we will complement the status dependent system, with a set of adequate probability equations. This approach was introduced by Chakrabarti et.al~\cite{ref-8}, who analyzed the  infection in the network using a system of probability equations, referred to as the Non-Linear Dynamical System (NLDS) model. Adopting their approach to the SIS process (\ref{eq:stochastic-SIS}), we obtain the following set of difference equations for the probabilities of states S and I:
\begin{equation}
\begin{array}{lll}
p_i^S(t+1) &=& p_i^S (t) (1 - f_i(t)) + \delta p_i^I(t) \\
p_i^{I}(t+1) &=& p^S_{i}(t) f_i(t) + (1 - \delta) p_i^I(t)
\label{eq:det-SIS}
\end{array}
\end{equation}
where $f_i(t)$ is now
\begin{equation}
f_i(t) = 1 - \prod_{j=1}^N (1 - \beta r_{ij} p_j^I(t)).
\label{eq:fi-det}
\end{equation}
Note that (\ref{eq:det-SIS}) is a deterministic equation.

Since $p_i^S(t) + p_i^I(t) = 1$ for all $i$ and all $t$, we rewrite (\ref{eq:det-SIS}) using $x_i = p_i^I$:
\begin{equation}
x_i(t+1) = (1 - x_i(t))f_i(t) + (1-\delta) x_i(t).
\label{eq:xi-SIS}
\end{equation}
Equation (\ref{eq:xi-SIS}) represents a nonlinear dynamical system $F : [0,1]^N \to [0,1]^N$.  The system (\ref{eq:xi-SIS}) has two fixed points: the origin $x_i = 0, \forall i \in \{1, \ldots, N \}$ and let $x_i^*(G)$ be the fixed point of (\ref{eq:xi-SIS}) different from the origin for the graph $G$.  
We will write only $x^*_i$ instead of $x^*_i(G)$ when it is clear which graph $G$ is considered in the context. At the stationary state: 
\begin{equation}
\delta x^*_i = (1 - x^*_i) \left[ 1 - \prod_{j=1}^N (1 - \beta r_{ij} x^*_j)  \right]. 
\label{eq:xi-stat}
\end{equation}

The origin $x_i = 0, \forall i \in \{1, \ldots, N \}$ is a fixed point of the system. Using the Jacobian matrix of the system (\ref{eq:xi-SIS}) evaluated at the origin:
$$
DF|_{(x_i = 0)} = (1 - \delta) I + \beta R,
$$
where $R=[r_{ij}]_{N \times N}$, one finds the well-known result~\cite{ref-8,Gomez-2010} that the origin is stable when 
\begin{equation}
\frac{\beta}{\delta} < \frac{1}{\lambda_{1,R} },
\label{eq:stability-cond}
\end{equation} 
where $\lambda_{1,R}$ is the largest eigenvalue of the matrix $R$. Whenever the infection to cure ratio $\beta / \delta$ is greater than the network threshold $1 / \lambda_{1,R}$ the disease will reach an endemic state in the network. For a contact process $\lambda_{1,R}=1$, since $R$ is a row stochastic matrix, while for a reactive process $\lambda_{1,R} = \lambda_{1,A}$.  
Moreover, when $\beta \neq 0$, $\delta \neq 0$, and $\delta \neq 1$,
the ergodicity of the Markov chains describing the SIS dynamics of each node is guaranteed and therefore (\ref{eq:det-SIS}) has a unique globally stable fixed point. Therefore, there exists a critical value of $\beta$, $\beta_{c}=\delta/\lambda_{1,R}$, such that the origin is a globally asymptotically stable fixed point of (\ref{eq:xi-SIS}) if $\beta < \beta_{c}$, and $x_i^*$ for all $i$ is a globally asymptotically stable fixed point of (\ref{eq:xi-SIS}) when $\beta > \beta_{c}$.

\section{Reactive process} 
\label{sec:3}

\subsection{Upper bounds on the probability of being infective} 

In this section we consider a family $\Phi$ of all possible simple and connected graphs with at least two nodes (we exclude from this family the empty graph and the graph with a single node and no links) and SIS reactive processes on this family for which the stationary solution (\ref{eq:xi-stat}), different from the origin, is an asymptotically stable fixed point of (\ref{eq:xi-SIS}). For the reactive process,  since $r_{ij}=a_{ij}$, we rewrite (\ref{eq:xi-stat}) as:
\begin{equation}
\label{eq:xi-reactive}
x_i^* = \frac{\left[ 1 - \prod_{j=1}^N (1 - \beta a_{ij} x_j^*)  \right]}{\left[ 1 - \prod_{j=1}^N (1 - \beta a_{ij} x_j^*)  \right] + \delta}
\end{equation}

Our first observation which acted as a building block for deriving the bounds for contact and reactive processes was that the stationary probability of infection of the reactive model (\ref{eq:xi-reactive}) for all nodes $i$, is bounded by
\begin{equation}
\label{eq:bound-reactive}
x_i^* < \frac{1}{1+\delta} \equiv u_i^0.
\end{equation}
This is formally stated in lemma \ref{lemma-main} in Appendix \ref{app-reactive}.


Note that the bound (\ref{eq:bound-reactive}) is independent of the specific network topology. Its right-hand side corresponds to the stationary solution (\ref{eq:xi-reactive}) for an infinitely large full-mesh graph.
Bound (\ref{eq:bound-reactive}) is rough and uses no information about the topology. A better bound can be obtained if one considers the degree of node $i$; in this case, we have:    
\begin{equation}
\label{eq:bound-reactive-1}
x_i^* < \frac{1 - \left [ 1 - \frac{\beta}{1 + \delta} \right ]^{k_i}}{1 - \left [ 1 - \frac{\beta}{1 + \delta} \right ]^{k_i}+\delta} \equiv u_i^1 < u_i^0
\end{equation}
where $k_i$ is the degree of node $i$. 
In general, one can find progressively better bounds for $x_i^*$ by using more information on the graph topology. In fact, let 
\begin{equation} \label{eq:rec-for-u}
 u_i^n  =  \frac{1 - \prod_{j=1}^N (1 - \beta a_{ij} u_j^{n-1})}{1 - \prod_{j=1}^N (1 - \beta a_{ij} u_j^{n-1}) + \delta}
\end{equation}
where
\begin{equation*}
u_i^0 = 1/(1+\delta).
\end{equation*}
Then $x_i^*$ is bounded by
\begin{equation}
\label{eq:bound-reactive-all}
x_i^* < \ldots < u_i^{n} < \ldots < u_i^1 < u_i^0
\end{equation}
for all $i$. For a formal definition and proof see theorem \ref{theorem-reactive-all} in Appendix \ref{app-reactive}.

Using the similar arguments as in the proof of the theorem (\ref{theorem-reactive-all}), it can be shown that $\lim_{n\to \infty} u_i^n = x_i^*$ for all $i$. In this paper we are interested only for small $n$. A similar theorem to the theorem (\ref{theorem-reactive-all}) can also be proved for lower bounds but only for those SIS processes for which $\beta > \delta$. The obvious lower bound is $x_i^*>0$, but replacing 0 in a recurrent relation similar to the one in (\ref{eq:rec-for-u}) will produce only 0s.  
Appendix \ref{app-reactive} contains the theorem (\ref{theorem-reactive-all-1}) for lower bounds of $x_i^*$:   
\begin{equation}
\label{eq:bound-reactive-all-2}
L_i^0 \leq L_i^1 \leq \ldots \leq L_i^{n} < \ldots \leq  x_i^* 
\end{equation}
for all $i$, which is analogous to theorem (\ref{theorem-reactive-all}), and the bounds $L_i^n$ are defined as
\begin{equation*}
L_i^n  =  \frac{1 - \prod_{j=1}^N (1 - \beta a_{ij} L_j^{n-1})}{1 - \prod_{j=1}^N (1 - \beta a_{ij} L_j^{n-1}) + \delta} 
\end{equation*}
where
\begin{equation*}
L_i^0 = 1 - \delta/\beta.
\end{equation*}
Note that the left-hand side of (\ref{eq:bound-reactive-all-2}) is defined only for $\beta>\delta$, since $x_i^*>0$. This property comes from (\ref{eq:stability-cond}) since the graph associated with $L_i^0=x_{min}^*$ is a path graph of size $2$ with $\lambda_{1,G}=1$. $L_i^n$ for all $n$ are also going to be defined only for $\beta > \delta$, since the a priori assumption is that the peripheral nodes have no probability of being infected. In order to obtain bounds for $\beta<\delta$, we take a different approach described in the following subsection.


\subsection{Lower bounds on the probability of being infective} 

In the previous section we have derived upper bounds which are valid for all $\beta$ and $\delta$ and lower bounds valid only for the SIS processes for which $\beta > \delta$. Since this is a restriction, in this section we find lower bounds valid for all $\beta$ and $\delta$ by observing that if $G'=(V',E')$ is a subgraph of $G=(V,E)$, with $x_i^{*}$ and $x_i^{*'}$ being the stationary solution of (\ref{eq:xi-reactive}) for the graph $G$ and $G'$ respectively for an arbitrary node $i \in V\cap V'$, then $x_i^{*'} < x_i^{*}$. In other words, as we remove edges (and nodes) from a graph, the probability of infection will decrease for each (remained) node. This is stated formally in lemma \ref{lemma-reactive-lower} in Appendix \ref{app-reactive}.
Using this interesting property, we can derive lower bound for an arbitrary node by simply obtaining (numerically) the stationary solution of (\ref{eq:xi-reactive}) for the 1-hop neighborhood starting at node $i$. Then $x_i^{*'}$ is a lower bound for $x_i^{*}$:
\begin{equation}
 \label{eq:bound-reactive-lower}
 l_i^1 \equiv x_i^{*'}\leq x_i^*
\end{equation}
where
\begin{equation}
 x_i^{*'} = \frac{1 - \left[ 1 - \frac{\beta^2 x_i^{*'}}{\beta x_i^{*'} + \delta} \right]^{k_i}}{1 - \left[ 1 - \frac{\beta^2 x_i^{*'}}{\beta x_i^{*'} + \delta} \right]^{k_i} + \delta}
\end{equation}
is the stationary solution of (\ref{eq:xi-reactive}) for the hub (central node) of a star graph $G'$ with $k_i+1$ nodes, $k_i$ being the degree of node $i$.


Note that bound (\ref{eq:bound-reactive-lower}) unlike bound (\ref{eq:bound-reactive}) uses 1-hop topology information (the degree of the node) a priori, thus avoiding the problem when $\beta<\delta$. Consequently, one can find progressively better lower bounds for node $i$ by solving (\ref{eq:xi-reactive}) for different subgraphs of $G$.

To show this, we now define a class of subgraphs called a $p$-hop neighborhood.  
Let $i$ be an arbitrary node of the graph $G=(V,E)$, $i \in V$, and let $n_i = \max_{x} l(i,x)$ where $l(i,j)$ is the length of the shortest path between nodes $i$ and $j$. Let $V_i^0=\{ i \}$. We define a subgraph $G_i^p= (V_i^p, E_i^p)$ of $G=(V,E)$ as follows: 
\begin{eqnarray*}
V_i^p &=& \{ x | x\in V, 0\leq l(i,x) \leq p \} \\
E_i^p &=& \{ (x,y)| (x,y) \in E, x\in V_i^p, y\in V_i^{p-1} \},  
\end{eqnarray*}
where $p=1, \ldots, n_i + 1$. We say that $G_i^p$ is a $p$-hop neighborhood of node $i\in V$. (see Figure \ref{fig:subgraphs}). 
For example, $G_i^1=(\{i\}\cup V_i ,E_i)$, where $E_i$ is the set of edges adjacent to node $i$, and $V_i$ is the set of all neighbors of $i$. In fact, $G_i^1$ is a star graph with $k_i$ leaves and root $i$. Note that $G_i^{n_i + 1}$ is the entire graph $G$ and that the first triangle can occur in $G_i^2$ but not in $G_i^1$.

\begin{figure}
\centering
\includegraphics[scale = 0.45]{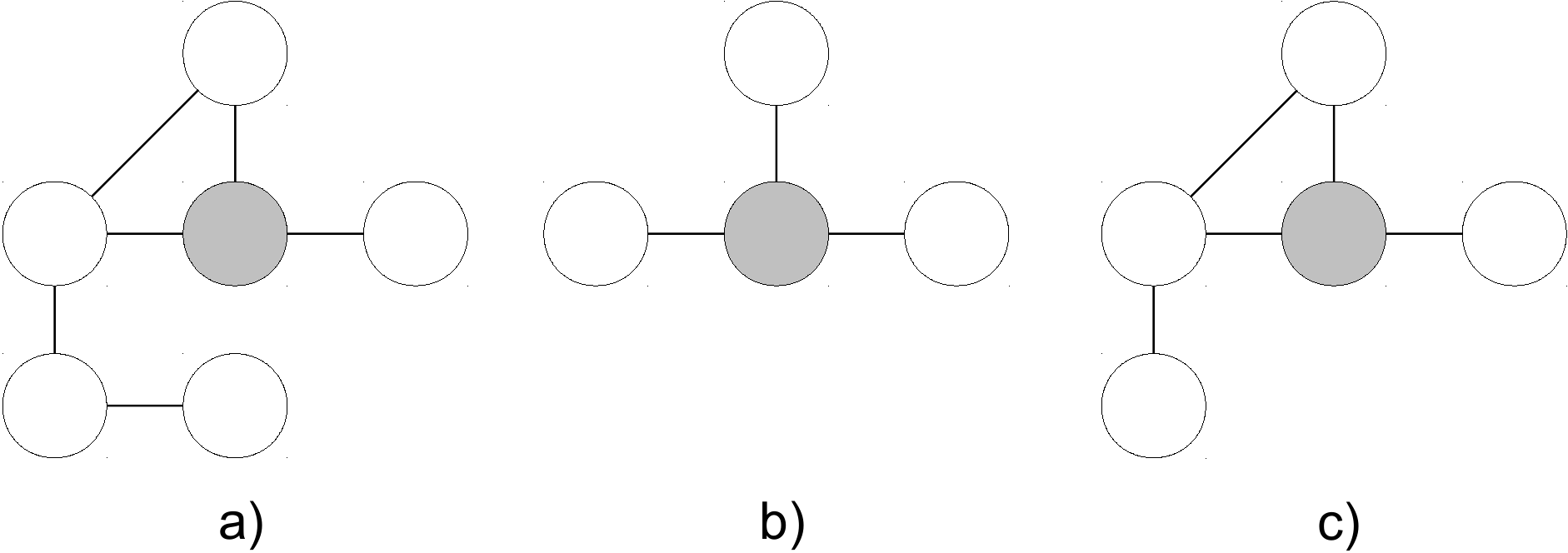}
\caption{b) and c) $1$-hop and $2$-hop neighborhood of the gray node extracted from the graph in a)}
\label{fig:subgraphs}
\end{figure}
Finally, if $l_i^p$ is the probability of infection of node $i$ given its $p$-hop neighborhood, by theorem \ref{theorem-reactive-lower-all} proven in Appendix \ref{app-reactive}, the probability of infection $x_i^*$ given the entire graph $G$ is bounded by
\begin{equation}
 \label{eq:bound-reactive-lower-all}
l_i^1 < l_i^2 < \ldots < l_i^{n_i+1} = x_i^*
\end{equation}
where $n_i+1$ is such that $E_i^{n_i + 1}=E$, i.e. the $n_i + 1$-hop neighborhood of node $i$ contains the entire graph $G$.

\subsection{Numerical results}

In the previous section we have proved that 
$$
l_i^p \leq x_i^* < u_i^n 
$$
for $i=1, \ldots N$, $p=1, \ldots n_i+1$, and $n=1,2, \ldots $. Note that only when $p=n_i+1$,  $l_i^p = x_i^*$; otherwise $l_i^p < x_i^*$. 
The bounds $l_i^1$ and $u_i^1$ are obtained by considering only (first) neighbors of $i$. The bound $u_i^1$ depends on the degree of the node $i$, that is, the information contained in the 1-hop neighborhood of $G$ extracted by starting at node $i$, while for the bound $l_i^1$ one computes the SIS model on the subgraph $G_i^1$, which is the subgraph of neighbors of $i$. 
In a similar fashion, the bounds $l_i^2$ and $u_i^2$ are obtained by considering second neighbors of $i$ (neighbors of the first neighbors). 
The bound $u_i^2$ can be computed by using $u_j^1$ for all neighbors $j$ of $i$.  Thus, $u_i^2$ reflects the topology of 2-hop neighborhood of $G$ extracted by starting at node $i$. Finally, for $n=n_i +1$, since $G_i^{n_i + 1}$ is the entire graph $G$, $u_i^{n_i+1}$ takes into account the topology of the whole network. 
Therefore, it makes sense to calculate the difference $d_i^p = u_i^p - l_i^p$, when $n=p$, and for small values of $p$.   
In this way, one could, at least numerically, answer one of the basic questions in mathematical epidemiology for any graph: what is the influence of the graph topology on disease spreading, or more precisely, on the probability that given node will be infected?

Lower bounds are derived as stationary solution of the SIS process for the corresponding subgraphs. On the other hand, upper bounds are found by back-propagation using the equation (\ref{eq:rec-for-u}). As a consequence,  when $p=n_i+1$, $l_i^{n_i+1} =x_i^*$ while $u_i^{n_i+1} > x_i^*$ and thus $d_i^{n_i+1} > 0$. 
In fact, see remark 3.3, only when $\lim_{n \to \infty} \left(  u_i^n \right) - l_i^{n_i+1}   = 0$.

In this section we study the Enron e-mail network obtained from \cite{jure-networkdata}, running (\ref{eq:xi-SIS}) on the network. The Enron e-mail network has 33696 nodes and 361622 edges with $\lambda_{1,A} = 118.4177$, and $\beta_{c} =  0.004222$ when 
$\delta=0.5$ for the reactive process. We study the upper and lower bounds of the expected density of infection $\rho = \sum_i x_i^* / N$ calculated as
\begin{eqnarray*}
  \hat{\rho_p} = \sum_{i=1}^N \frac{u_i^p}{N} \\
  \breve{\rho_p} = \sum_{i=1}^N \frac{l_i^p}{N} \\
\end{eqnarray*}
for different values of $p$, as well as the average difference $\Delta \rho_p$ between the upper and lower bound for all nodes, 
$$
\Delta \rho_p = \sum_{i=1}^N \frac{d_i^p}{N}.
$$
We also calculate $d_i^p$ for 3 nodes: the node with minimum degree, the node with maximum degree, and a node with average degree. To have a better idea of how much local information is being used, Table \ref{table:avgsize} depicts the size of a $p$-hop neighborhood for the Enron e-mail network, $\left| E^p \right|$ , as measured by the number of edges in the corresponding subgraph averaged over all nodes $i$, as well as its fraction of the total number of edges in the network.

\begin{table}[htbp]
\caption{Average size of $p$-hop neighborhood for the Enron e-mail network. $\left| E^p \right|$ is an average of $\left| E_i^p \right|$ over all nodes $i$ and $\left| E \right|$ is the total number of edges in the network.}
\label{table:avgsize} \centering
\begin{ruledtabular}
\begin{tabular}{ccc}
$p$ & $\left| E^p \right|$ & $\left| E^p \right| / \left| E \right|$  \\ \toprule[0.1em]
1 & 10 & 0.0003 \\ 
2 & 1538 & 0.004 \\
3 & 45067 & 0.125 \\
4 & 207496 & 0.574 \\ 
\end{tabular}
\end{ruledtabular}
\end{table}

\begin{figure}[htbp]
\centering
\includegraphics[scale = 0.47]{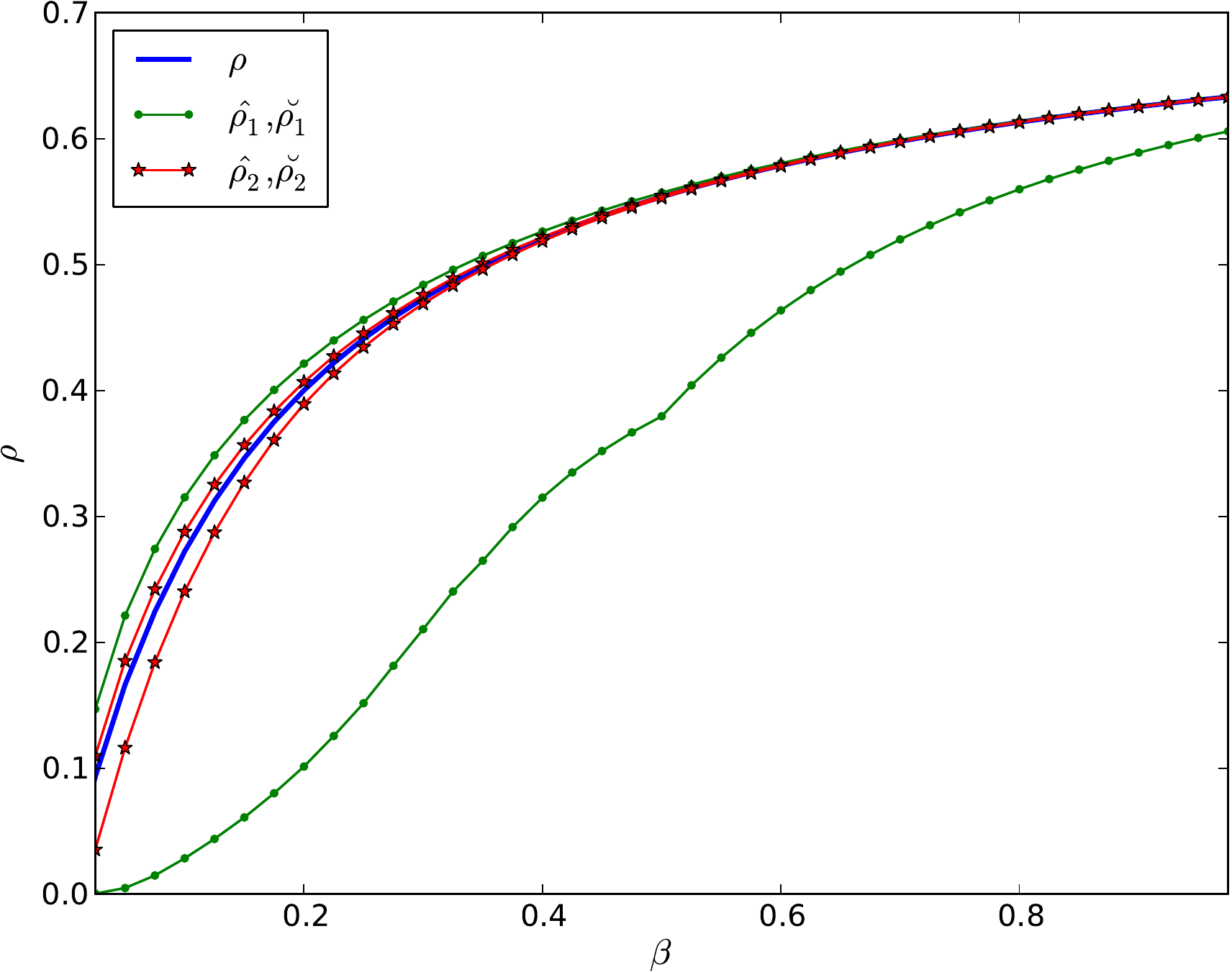}
\caption{The density of infective nodes in the Enron e-mail network as the transmission parameter $\beta$ is varied, and $\delta = 0.5$, obtained by simulating (\ref{eq:xi-reactive}) on the network until the model stabilizes, along with the upper and lower bounds, $\hat{\rho_p}$ and $\breve{\rho_p}$, on $\rho$ using  1-hop and 2-hop topology information.}
\label{fig:reactive}
\end{figure}

Figures \ref{fig:reactive} and \ref{fig:gap-reactive} summarize our results.
As depicted in Figure \ref{fig:reactive}, the bounds are surprisingly tight even when only 2-hop topology information is being used. More precisely, we obtain a very good estimate for the infection density $\rho$ by summing over node-level estimates which are using only $0.4\%$ of the network topology on average (see Table \ref{table:avgsize}). The upper bound $\hat{\rho_1}$ is also surprisingly tight given that it only uses the node's degrees while having no information for the edges in the network.

Figure \ref{fig:gap-reactive} shows that the average difference $\Delta \rho_p$ between the bounds decreases as one considers 1-hop, 2-hop, and 3-hop topology information, as expected since the bounds around the stationary infection density $\rho = \sum_i x_i^* / N$ become tighter. Also the spreading becomes less topology dependent as the disease transmission parameter $\beta$ increases. 
When $\beta$ is close to $\beta_c$, $x_i^*$ and consequently $\rho$, are close to zero as well. Therefore, as $\beta$ approaches $\beta_c$, it is expected that the value of $x_i^*$ is influenced by the whole network. 
For the Enron network, as indicated in Figure \ref{fig:gap-reactive}, the average difference between the bounds calculated from 3-hop neighborhood is close to zero, $\Delta \rho_3 \leq 0.012$ for all $\beta$. Additionally, when $\beta > 0.4$, topology of only 2-hops away is relevant for the spreading process, $\Delta \rho_2 \leq 0.01$. The bounds calculated using only 1-hop topology, i.e. the nodes' degrees are wide apart for all values of $\beta$ for the particular network, indicating that the specific degrees are highly influential in the spreading process.

In this light, we examine the difference $d_i^p$ between the bounds for three randomly chosen nodes with particular degree: one with minimum degree, one with average degree, and one with maximum degree. Note that results vary greatly for the three types of nodes. The difference $d_i^p$ is smallest for the node with minimum degree reaching a maximum of 0.028 for all $\beta$ given only 2-hop topology information (0.4\% of the entire network topology on average). It is also worth noting that the lower bound is very close to the actual result, with the difference being due to the upper bound requiring more information to converge.
Interestingly, while the gap for the node with maximum degree quickly decreases with the increase of $\beta$, it can be large for specific values of $\beta$. In contrast to the minimum degree case, for the node with maximum degree, the difference is due to the lower bound not having converged, while the upper bound is quite tight. Another interesting result for the node with maximum degree is that as $\beta$ gets greater than 0.04, knowing only the neighbors of the node's neighbors suffices for predicting the outcome of the infection.
For the node with average degree, we observe two interesting results. Firstly, the difference $d_i^p$ exists for a relatively wide span of $\beta$ (as in the minimum degree case). The other result is that for some specific values of $\beta$, $d_i^p$ can be relatively large (as in the maximum degree case). However, the difference $d_i^p$ is smaller than 0.03 for all values of $\beta$ given the 3-hop topology information which constitutes approximately 12.5\% of the total network topology on average.

%
Finally, for all nodes, the bounds on the probability of being infective are tighter as $\beta \to 1$, and looser as $\beta \to \beta_c$. The conclusion from this is that as $\beta \to \beta_c$, network topology plays a bigger role in the dynamics of the spreading process.

\begin{figure*}[htbp]
\centering
\includegraphics[scale = 0.5]{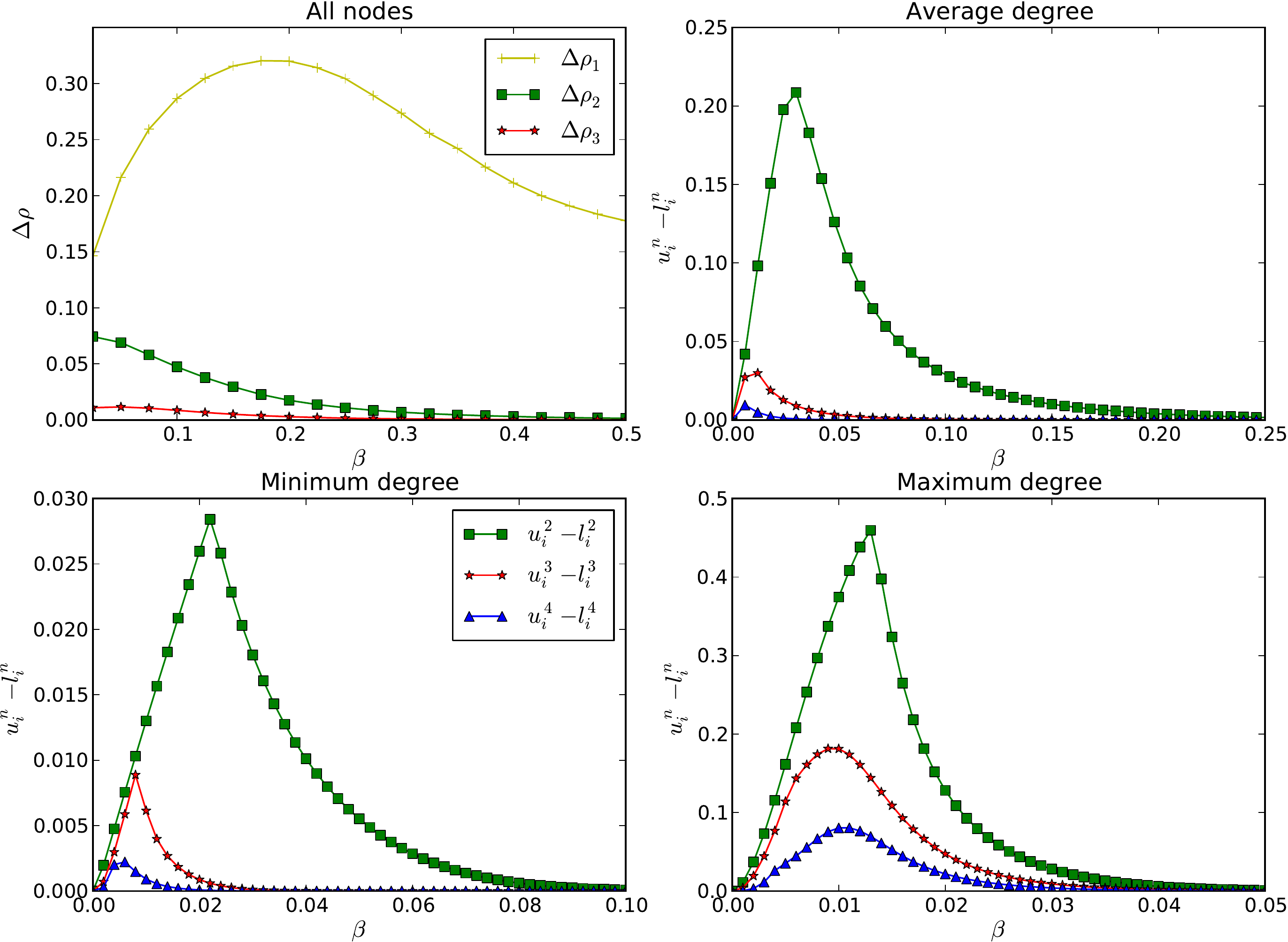}
\caption{The average difference $\Delta \rho_p$ between the upper and lower bound (top left) and the difference between the upper and lower bounds that use 2-hop, 3-hop and 4-hop topology information for a node $i$ in the Enron e-mail network that has minimum (bottom left), average (top right) and maximum degree (bottom right) as the transmission parameter $\beta$ is varied and $\delta = 0.5$.}
\label{fig:gap-reactive}
\end{figure*}


\section{Contact process}
\label{sec:4}

\subsection{Bounds on the probability of being infective}

For the contact process, since $r_{ij}=\frac{a_{ij}}{\sum_k{a_{ik}}}$, we rewrite (\ref{eq:xi-stat}) as:
\begin{equation}
\label{eq:xi-contact}
x_i^* = \frac{\left[ 1 - \prod_{j=1}^N (1 - \beta \frac{a_{ij}x_j^*}{\sum_k a_{ik}})  \right]}{\left[ 1 - \prod_{j=1}^N (1 - \beta \frac{a_{ij}x_j^*}{\sum_k a_{ik}})  \right] + \delta}
\end{equation}

From lemma \ref{lemma-contact} and \ref{lemma-contact2} (see Appendix \ref{app-contact}) we have the following bounds for the stationary solution of (\ref{eq:xi-contact})
\begin{equation}
\label{eq:bound-contact}
l_i^0 \equiv \frac{1 - e^{-\beta x_i^*}}{1 - e^{-\beta x_i^*} + \delta} < x_i^* \leq 1 - \frac{\delta}{\beta} \equiv u_i^0
\end{equation}
Note that the bound (\ref{eq:bound-contact}) is independent of the specific network topology. Its left-hand and right-hand side correspond to the stationary solution of (\ref{eq:xi-contact}) for an infinitely large full-mesh graph and a path graph of size 2 respectively.

Similarly to the reactive process, better bounds can be obtained if we use a node's degree:
\begin{eqnarray}
\label{eq:bound-contact-1}
l_i^1 & \equiv & \frac{1 - \left [ 1 - \beta x_p \right ]^{k_i}}{1 - \left [ 1 - \beta x_p \right ]^{k_i}+\delta} < x_i^* \leq \nonumber \\
 & \leq & \frac{1 - \left [ 1 - \beta + \delta \right ]^{k_i}}{1 - \left [ 1 - \beta + \delta \right ]^{k_i}+\delta} \equiv u_i^1
\end{eqnarray}
where $k_i$ is the degree of node $i$ and $x_p$ is the solution of the equation:
\begin{equation*}
x=\frac{1 - e^{-\beta x}}{1 - e^{-\beta x} + \delta}
\end{equation*}
More generally, by using $n$-hop neighborhoods, the stationary solution of (\ref{eq:xi-contact}) for an arbitrary node $i$, $x_i^*$, is bounded by,
\begin{equation}
\label{eq:bound-contact-all}
l_i^0 < l_i^1 < \ldots < l_i^n  <  x_i^* \leq u_i^n \leq \ldots \leq u_i^1 \leq u_i^0
\end{equation}
where
\begin{eqnarray*}
 l_i^n  =  \frac{1 - \prod_{j=1}^N (1 - \beta r_{ij} l_j^{n-1})}{1 - \prod_{j=1}^N (1 - \beta r_{ij} l_j^{n-1}) + \delta} \mbox{ and } \\
 u_i^n  =  \frac{1 - \prod_{j=1}^N (1 - \beta r_{ij} u_j^{n-1})}{1 - \prod_{j=1}^N (1 - \beta r_{ij} u_j^{n-1}) + \delta}
\end{eqnarray*}
and
\begin{equation*}
l_i^0 = \frac{1 - e^{-\beta x_i^*}}{1 - e^{-\beta x_i^*} + \delta} \mbox{ and } u_i^0 = 1 - \frac{\delta}{\beta}
\end{equation*}
For a formal definition and proof see theorem \ref{theorem-contact-all} in Appendix \ref{app-contact}. Note that here, unlike in the reactive process, the problem with $u_i^0 = 1 - \delta/\beta$ when $\delta>\beta$ is avoided since $\beta_c > \delta$.
\subsection{Numerical results} 

Again, it makes sense to calculate the difference $d_i^n = u_i^n - l_i^n$ between the upper and lower bound derived by using n-hop topology information for node $i$ to determine the dependence of the contact process on the specific network topology. 

Figures \ref{fig:contact} and \ref{fig:gap-contact} show the numerical results. As in the reactive process, we study the upper bound $\hat{\rho_p}$ and lower bound $\breve{\rho_p}$ on $\rho$, as well as the average difference between the upper and lower bounds $\Delta \rho_p$. We also study $d_i^p$ for 3 nodes: the node with minimum degree, a node with average degree, and the node with maximum degree. Note that $\beta_c = 0.5$ when $\delta = 0.5$. 

In general, the contact process is less dependent on the network topology than the reactive process, as the largest value of $\Delta \rho_1$ is an order of magnitude less than the corresponding value for the reactive process. The more topology information is included in the calculation of the bounds, the difference between them decreases. Also, the probabilities that each of the 3 examined nodes is infective are equally dependent on the network topology, since $d_i^n$ for $n = 1,2,3$ are similarly valued. Contrary to the reactive process, the bounds on the probability of being infective are tighter as $\beta \to \beta_c$ (see Figure \ref{fig:contact}). In the Appendix \ref{app-limit}, we show that when $\beta$ is close to $\beta_{c}= \delta$, the probabilities of being infective have an analytical solution in closed form, they are no longer topology dependent, and are functions only of the spreading process parameters $\beta$ and $\delta$.

\begin{figure}[htbp]
\centering
\includegraphics[scale = 0.45]{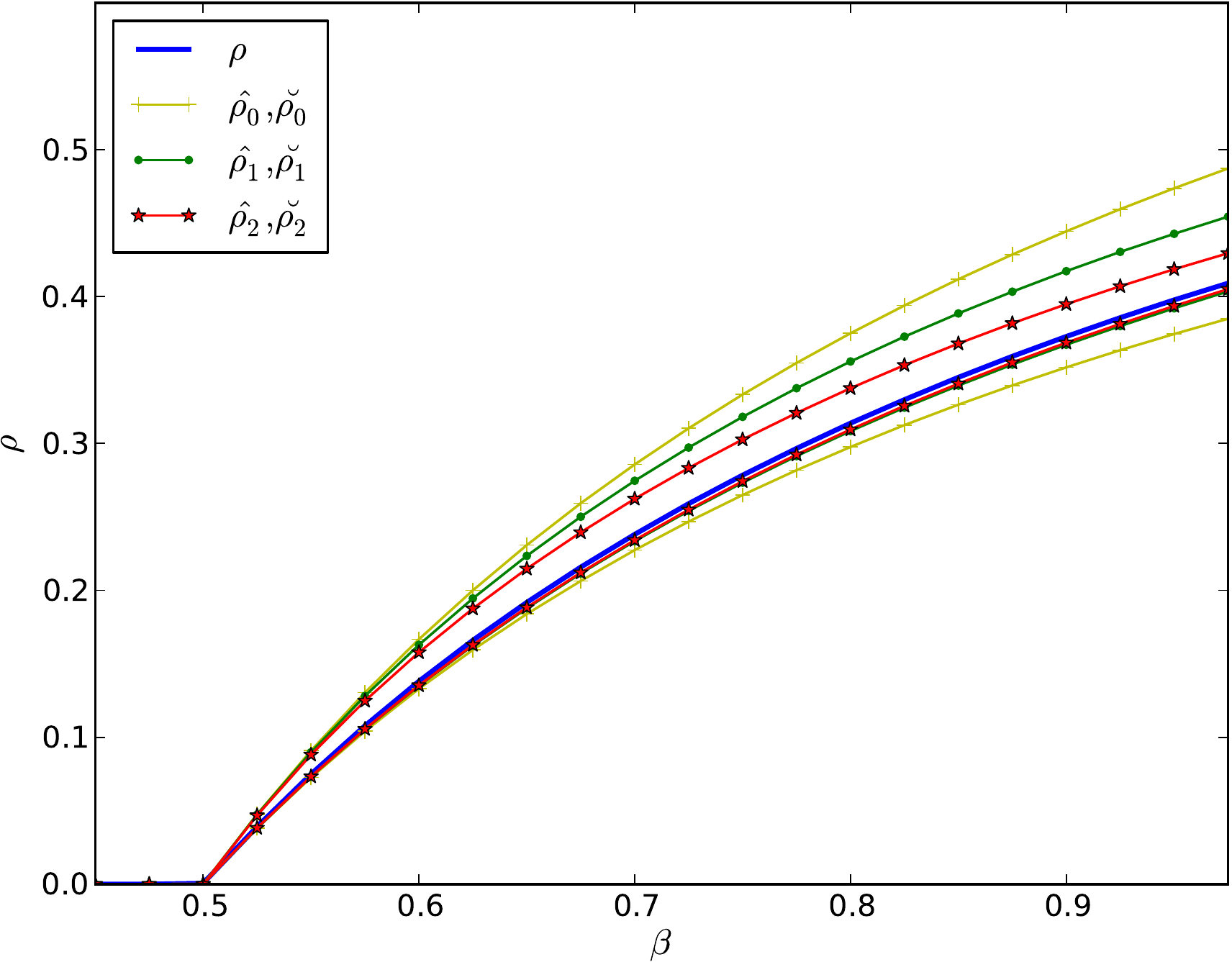}
\caption{The expected infection density $\rho$ in the endemic state for the Enron e-mail network for the contact process as $\beta$ is varied, and $\delta=0.5$. The bounds on $\rho$ calculated with no topology information $\hat{\rho_0}$, $\breve{\rho_0}$, 1-hop topology information $\hat{\rho_1}$, $\breve{\rho_1}$, and 2-hop topology information $\hat{\rho_2}$, $\breve{\rho_2}$ are depicted as well.}
\label{fig:contact}
\end{figure}

\begin{figure*}[htbp]
\centering
\includegraphics[scale = 0.4]{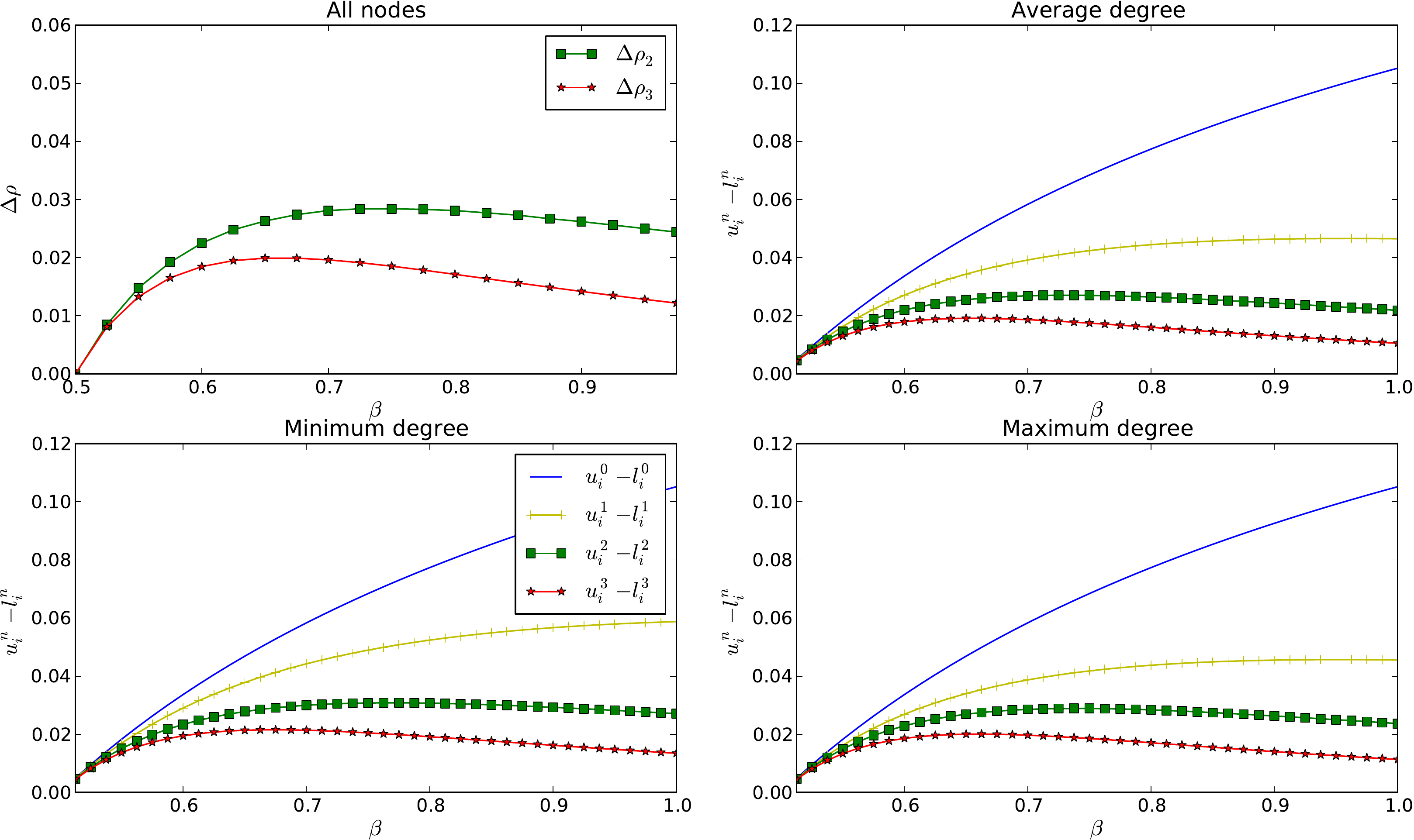}
\caption{The average difference $\Delta \rho_n$ between the upper and lower bound (top left) and the difference between the upper and lower bounds that use 0-hop, 1-hop, 2-hop and 3-hop topology information for a node $i$ in the Enron e-mail network that has minimum (bottom left), average (top right) and maximum degree (bottom right) as the transmission parameter $\beta$ is varied and $\delta = 0.5$.}
\label{fig:gap-contact}
\end{figure*}

\section{Conclusions}  
\label{sec:conc}

In this paper we have derived the upper and lower bounds on the probability that a node is infective for the SIS model of infection spreading on networks, where the behavior of a node is modeled with a discrete-time Markov chain SIS model. We have considered the reactive and the contact process as two cases of the spreading process. For both processes we use the difference between upper and lower bounds on microscopic level to assess the dependence of the spreading processes on network topology. Numerical results are given on the Enron e-mail network. For both processes, the bounds are progressively better as one considers a larger $n-$hop neighborhood of a node. For the reactive process, both bounds on the probability that a node is infective are tighter as $\beta \to 1$ and their difference is largest for nodes with average degree. Conversely, the bounds on the probability that a node is infective for the contact process are tighter as $\beta \to \beta_c$. 

One  of the main implications of the paper is that if $\beta$ is larger than its critical value (when $\beta$ is close to $\beta_c$ the probability of a node to be infective is anyway close to zero), one can estimate the probability of being infective using only local information (considering only $n-$hop local topology, for small $n$), without knowing the whole network. Consequently, from this local information one can also estimate the density of being infective on the whole network, as well as assess the extend to which the topology affects the outcome of the infection on macroscopic level.

The results of this paper are easily extendable to other ergodic models (such as SIRS, for example) and are related to all types of spreading (idea, failure, rumor) \cite{rumor-model, GoffmanNewill, infoblogspace}, regardless on the type of the spread agent. How these results can be extended to SIR model by considering SIRS model and taking one of its parameters to approach zero (or one) so that SIRS model in this limit approaches SIR model is a question for further research.     

\appendix
\section{Bounds for reactive process}
\label{app-reactive}

\begin{lemma}
\label{lemma-main}
Let $\Phi$ be the family of all possible simple and connected graphs $G=(V,E)$ with $\left| V \right| \geq 2$. Let $\mathbf{x}^*(G)=[x_1^* x_2^* \ldots x_N^*]$ be the stationary solution (\ref{eq:xi-stat}) different from the origin for the graph $G \in \Phi$, $G=(V,E)$, where $N=\left|V\right|$. 
Let $x_{max}^*=\max_{G\in \Phi}\max_i{\mathbf{x}^*(G)}$. Then for all $i$, $x_i^*$ is bounded by
\begin{equation*}
x_i^* < \frac{1}{1+\delta} \equiv u_i^0.
\end{equation*}
\end{lemma}

\begin{proof}
Let $\Delta$ be the set of neighbors of the node associated with the value $x_{max}^*$. 
We will show that $x_j^*=x_{max}^*$ for all $j\in \Delta$ by using contradiction. Let the node associated with the value $x_{max}^*$ be node $k$, i.e., $x_k^*=x_{max}^*$. Assume that $x_j^*=x_{max}^*$ for all $j\in \Delta$ is false. Then, there exists at least one node $i\in \Delta$ such that $x_i^*<x_{max}^*$. But, this means that $x_k^*<x_{max}^*$ since  
\begin{equation*}
 \frac{\partial x_k}{\partial x_i}=\frac{\beta\delta r_{ki}}{(f_k + \delta)^2}\prod_{j\in \Delta \setminus \{i\}}{\left(1-\beta r_{kj}x_j\right)}>0
\end{equation*}
which contradicts our first statement that $x_k^*=x_{max}^*$. 
Let $n=\left|\Delta\right|$. From  (\ref{eq:xi-reactive}) and the fact that $x_j^*=x_{max}^*$ for all $j\in \Delta$, we have:
\begin{equation*}
 x_{max}^* = \frac{1 - \left(1 - \beta x_{max}^*\right)^{n}}{1 - \left(1 - \beta x_{max}^*\right)^{n} + \delta}
\end{equation*}
Since $\frac{\partial x_{max}^*}{\partial n} > 0$, the maximum value of $x_{max}^*$ is obtained at $n \to\infty$. 
Finally, the bound (\ref{eq:bound-reactive}) comes directly from 
 \begin{equation*}
 x_{max}^* = \lim_{n\to\infty} \frac{1 - \left(1 - \beta x_{max}^*\right)^{n}}{1 - \left(1 - \beta x_{max}^*\right)^{n} + \delta} = \frac{1}{1+\delta} 
\end{equation*}
\end{proof}

\begin{theorem}
\label{theorem-reactive-all}
Let $\Phi$ be the family of all possible simple and connected graphs $G=(V,E)$ with $\left| V \right| \geq 2$. 
Let $\mathbf{x}^*(G)=[x_1^* x_2^* \ldots x_N^*]$ be the stationary solution  (\ref{eq:xi-reactive}) different from the origin and let $i$ be an arbitrary node of the graph $G=(V,E)$, $i \in V$.
Let 
\begin{equation*} 
 u_i^n  =  \frac{1 - \prod_{j=1}^N (1 - \beta a_{ij} u_j^{n-1})}{1 - \prod_{j=1}^N (1 - \beta a_{ij} u_j^{n-1}) + \delta}
\end{equation*}
where
\begin{equation*}
u_i^0 = 1/(1+\delta).
\end{equation*}
Then $x_i^*$ is bounded by
\begin{equation*}
x_i^* < \ldots < u_i^{n} < \ldots < u_i^1 < u_i^0
\end{equation*}
for all $i$.
\end{theorem}
\begin{proof}
We will first prove that $u_i^n < \ldots < u_i^1 < u_i^0$ by induction. Note that $u_i^0=u^0$ is topology independent, and $u^0 = \lim_{k_i \to \infty} u_i^1$ and since $\frac{\partial u_i^1}{\partial k_i}>0$ we have that $u_i^1 < u^0$ for all $i$. Now assume that $u_i^p < u_i^{p-1}$ holds for all $i$ and $p=2,3,\dots,n-1$. Since $\frac{\partial u_i^n}{\partial u_j^{n-1}}>0$ and $u_j^{n-1} < u_j^{n-2}$ it follows that $u_i^{n} < u_i^{n-1}$ for all $i$. We will prove that $x_i^* < u_i^n$ in a similar fashion. From Lemma \ref{lemma-main} we have that $x_i^* < u_i^0=u^0$ for all $i$. Now assume that $x_i^* < u_i^p$ holds for all $i$ and $p=1,2,\dots,n-1$. Note that $x_i^*$ is $u_i^n$ with $u_j^{n-1}$ replaced by the smaller $x_j^*$  ($x_j^* < u_j^{n-1}$). Since $\frac{\partial x_i^*}{\partial x_j^*}>0$ it follows that $x_i^* < u_i^n$ for all $i$. 
\end{proof}  

\begin{theorem}
\label{theorem-reactive-all-1}
Let $\Phi$ be the family of all possible simple and connected graphs $G=(V,E)$ with $\left| V \right| \geq 2$. Let $\beta > \delta$.  
Let $\mathbf{x}^*(G)=[x_1^* x_2^* \ldots x_N^*]$ be the stationary solution  (\ref{eq:xi-reactive}) different from the origin and let $i$ be an arbitrary node of the graph $G=(V,E)$, $i \in V$.
Let 
\begin{equation*}
L_i^n  =  \frac{1 - \prod_{j=1}^N (1 - \beta a_{ij} L_j^{n-1})}{1 - \prod_{j=1}^N (1 - \beta a_{ij} L_j^{n-1}) + \delta} 
\end{equation*}
where
\begin{equation*}
L_i^0 = 1 - \delta/\beta.
\end{equation*}
Then $x_i^*$ is bounded by
\begin{equation*}
L_i^0 \leq L_i^1 \leq \ldots \leq L_i^{n} < \ldots \leq  x_i^* 
\end{equation*}
for all $i$.
\end{theorem}

\begin{lemma}
\label{lemma-reactive-lower}
Let $G'=(V',E')$ be a subgraph of $G=(V,E)$. Let node $i \in V\cap V'$ and let $x_i^{*}$ and $x_i^{*'}$ be the stationary solution of (\ref{eq:xi-reactive}) different from the origin associated with the node $i$ for the graph $G$ and $G'$ respectively. Then $x_i^{*'} < x_i^{*}$.
\end{lemma}
\begin{proof}
 Without loss of generality, assume that only one edge $e$ between nodes $i$ and $j$ is removed from $G$ in order to obtain $G'$. Starting from the stationary solution $\textbf{x}^{*}(G)$, and from node $i$'s point of view, the edge removal can be interpreted as a change in $x_j$ from $x_j^*$ to $0$. Then, from Lemma \ref{lemma-main} we have a negative change in $x_i$ which will propagate and imply negative changes in $x_k$ for all $k \in V'$ with each iteration of (\ref{eq:xi-reactive}). On the other hand, removing a node can be interpreted as a removal of its edges.
\end{proof}

\begin{theorem}
\label{theorem-reactive-lower-all}
Consider an arbitrary node $i$ of the graph $G=(V,E)$ and let $G_i^p=(V_i^p,E_i^p)$ be the $p$-hop neighborhood of $G$ extracted by starting at node $i$. Let $n_i = \max_{x} l(i,x)$ and let $\mathbf{x}^*(G)=[x_1^* x_2^* \ldots x_N^*]$ and $\mathbf{x}(G_i^p)^*=[l_1^p l_2^p \ldots l_{n}^p]$ be the stationary solution of (\ref{eq:xi-reactive}) different from the origin for the graphs $G=(V,E)$ and $G_i^p=(V_i^p,E_i^p)$, respectively.  
Then $x_i^*$ is bounded by
\begin{equation*}
l_i^1 < l_i^2 < \ldots < l_i^{n_i+1} = x_i^*
\end{equation*}
for all $i\in V$.
\end{theorem}
\begin{proof}
  $l_i^{p-1} < l_i^{p}$ for all $p=2,3,\dots, n_i+1$ comes directly from Lemma \ref{lemma-reactive-lower}.
\end{proof}

\section{Bounds for reactive process}
\label{app-contact}
\begin{lemma}
\label{lemma-contact}
Let $\Phi$ be the family of all possible simple unweighted connected graphs. Let $\mathbf{x}^*(G)=[x_1^* x_2^* \ldots x_n^*]$ be the stationary solution (\ref{eq:xi-contact}) different from the origin for the graph $G \in \Phi$. Let $x_{min}^*=\min_{G\in \Phi}\min_i{\mathbf{x}^*(G)}$ and $x_{max}^*=\max_{G\in \Phi}\max_i{\mathbf{x}^*(G)}$. Let $\Gamma$ and $\Delta$ be the set of neighbors of the node associated with the value $x_{min}^*$ and $x_{max}^*$ respectively.  Then $\left|\Gamma\right|\to\infty$ and $\left|\Delta\right|=1$.
\end{lemma}
\begin{proof}
 Let $x_{lim}^*$ be either $x_{min}^*$ or $x_{max}^*$ and let the number of its neighbours be $n$. From Lemma \ref{lemma-main} we have:
\begin{equation*}
 x_{lim}^* = \frac{1 - \left(1 - \frac{\beta x_{lim}^*}{n}\right)^{n}}{1 - \left(1 - \frac{\beta x_{lim}^*}{n}\right)^{n} + \delta}
\end{equation*}
Since $\frac{\partial x_{lim}^*}{\partial n} < 0$, the minimum value of $x_{lim}^*$ is obtained at $n=\left|\Gamma\right|\to\infty$ and the maximum at $n=\left|\Delta\right|=1$ since the graph must be connected.
\end{proof}

\begin{lemma}
\label{lemma-contact2}
Let $\mathbf{x}^*(G)=[x_1^* x_2^* \ldots x_N^*]$ be the stationary solution of (\ref{eq:xi-contact}) different from the origin. Then $x_i^*$ is bounded by
\begin{equation*}
l_i^0 \equiv \frac{1 - e^{-\beta x_i^*}}{1 - e^{-\beta x_i^*} + \delta} < x_i^* \leq 1 - \frac{\delta}{\beta} \equiv u_i^0
\end{equation*}
for all $i$.
\end{lemma}
\begin{proof}
 From Lemma \ref{lemma-contact} the bounds come directly from the solution of the equations
 \begin{equation*}
 x_{min}^* = \lim_{n\to\infty} \frac{1 - \left(1 - \frac{\beta x_{min}^*}{n}\right)^{n}}{1 - \left(1 - \frac{\beta x_{min}^*}{n}\right)^{n} + \delta}
\end{equation*}
 \begin{equation*}
 x_{max}^* = \frac{\beta x_{max}^*}{\beta x_{max}^* + \delta}
\end{equation*}
\end{proof}

\begin{theorem}
\label{theorem-contact-all}
Let $\mathbf{x}^*(G)=[x_1^* x_2^* \ldots x_N^*]$ be the stationary solution of (\ref{eq:xi-contact}) different from the origin and let
\begin{eqnarray*}
 l_i^n  =  \frac{1 - \prod_{j=1}^N (1 - \beta r_{ij} l_j^{n-1})}{1 - \prod_{j=1}^N (1 - \beta r_{ij} l_j^{n-1}) + \delta} \mbox{ and } \\
 u_i^n  =  \frac{1 - \prod_{j=1}^N (1 - \beta r_{ij} u_j^{n-1})}{1 - \prod_{j=1}^N (1 - \beta r_{ij} u_j^{n-1}) + \delta}
\end{eqnarray*}
where
\begin{equation*}
l_i^0 = \frac{1 - e^{-\beta x_i^*}}{1 - e^{-\beta x_i^*} + \delta} \mbox{ and } u_i^0 = 1 - \frac{\delta}{\beta}
\end{equation*}
then $x_i^*$ is bounded by
\begin{equation*}
l_i^0 < l_i^1 < \ldots < l_i^n  <  x_i^* \leq u_i^n \leq \ldots \leq u_i^1 \leq u_i^0
\end{equation*}
for all $i$.
\end{theorem}
\begin{proof}
The proof is completely analogous to that of Theorem \ref{theorem-reactive-all}.
\end{proof} 

\section{Analytical solution for the contact process in the limit $\beta \to \beta_c$}
\label{app-limit}

When $\beta \rightarrow \beta_{c}$  (but $\beta > \beta_{c}$) then the probability $x_i^*$ that node $i$ is infective is $x_i^* \approx \varepsilon_i$, where $0 \leq \varepsilon_i \ll 1$, and from (\ref{eq:xi-stat}) (neglecting second order terms in $\varepsilon$) one gets
\begin{equation}
\delta \varepsilon_i = (1-\varepsilon_i) \beta \sum_j r_{ij} \varepsilon_j.
\label{eq:xi-epsilon}
\end{equation}
Let $y = [\varepsilon_1 \ldots \varepsilon_N]$ and $D_y=[d_{ij}]$ be a diagonal matrix such that $d_{ii} = 1 - \varepsilon_i$ and $d_{ij}=0$ for $i \neq j$.  The last equation can be written in matrix form as 
$$
\frac{\delta}{\beta} y = D_y R y, 
$$
or 
$$
\left[ D_y R - \frac{\delta}{\beta} I_N   \right] y =0.
$$
Assuming that $\varepsilon_i \neq 0$ for all $i$, the last equation reduces to 
$ D_y R - \frac{\delta}{\beta} I_N =0$, which, since $\sum_j r_{ij}= 1$, has a solution $\varepsilon_i = 1 - {\delta}/{\beta}$ for all $i$. 
%
Therefore, when $\beta > \beta_{c}= \delta$ and the nodes' probabilities $x_i^*$ of being infective are small, the $x_i^*$'s  have an analytical solution in closed form, they are no longer topology dependent, and are functions only of the spreading process parameters $\beta$ and $\delta$.


\begin{thebibliography}{99}











\bibitem{ref-3} M. E. J. Newman, 
Physical Review E 66, 016128 (2002).

\bibitem{ref-4} R. Pastor-Satorras and A. Vespignani, 
Physical Review Letters 86, 3200 (2001).

\bibitem{ref-5} R. Pastor-Satorras and A. Vespignani, 
Physical Review E 63, 066117 (2001).

\bibitem{ref-6} V. M. Eguíluz and K. Klemm, 
Physical Review Letters 89, 108701 (2002).

\bibitem{ref-7} Y. Wang,  D. Chakrabarti, C. Wang,  and C. Faloutsos, 
In Proc. Symp. Reliable and Distributed Systems, Florence, Italy, Oct. 2003.

\bibitem{ref-2008} M. Draief, A. Ganesh and L. Massoulie, Threshold for virus spread on Networks, Ann. Appl. Probab. Volume 18, Number 2, 359--378 (2008)


\bibitem{hmf1} M. Barthelemy, A. Barrat, R. Pastor-Satorras, and A. Vespignani, Phys. Rev. Lett. 92, 17870 (2004).

\bibitem{hmf2} J. Gomez-Gardenes, V. Latora, Y. Moreno, and E. Profumo, Proc. Nat. Acad. Sci. USA 105, 1399-1404 (2008).

\bibitem{hmf-accuracy} B. Guerra and J. Gomez-Gardenes, Phys. Rev. E. 82, 035101(R) (2010). 

\bibitem{ref-8} D. Chakrabarti, Y. Wang, C. Wang, J. Leskovec, and C. Faloutsos, Epidemic Thresholds in Real Networks, 
ACM Trans. Inf. Syst. secur. Vol. 10, 13, (2008)



\bibitem{Gomez-2010} S. Gomez, A. Arenas, J. Borge-Holthoefer, S. Meloni and Y. Moreno, Discrete-time Markov chain approach to contact-based disease
spreading in complex networks, EPL 89, 38009 (2010)


\bibitem{bounded-contact} S. Meloni, A. Arenas, and Y. Moreno, Proc. Nat. Acad. Sci. USA 106, 16897-16902 (2009). 

\bibitem{npHMF} S. Gómez, J. Gómez-Gardeñes, Y. Moreno and A. Arenas, 
Physical Review E 84, 036105 (2011).

\bibitem{ref-2011} M. Deijfen, Epidemics and vaccination on weighted graphs, 
Mathematical Biosciences, in press (2011) 

\bibitem{psv-SIS} C. Castellano and R. Pastor-Satorras, 
Physical Review Letters 105, 218701 (2010)

\bibitem{reinfection} R. Parshani, S. Carmi and S. Havlin, 
Physical Review Letters 104, 258701 (2010)

\bibitem{contact-1} J. Marro and R. Dickman, Nonequilibrium phase transitions
in lattice models, (Cambridge University Press, Cambridge) 1999.

\bibitem{contact-2} C. Castellano C. and R. Pastor-Satorras, Non-mean-
field behavior of the contact process on scale-free networks, Phys. Rev. Lett., 96 (2006) 038701.

\bibitem{contact-3} C. Castellano and R. Pastor-Satorras, Routes to
thermodynamic limit on scale-free networks, Phys. Rev. Lett., 100 (2008) 148701.


\bibitem{reactive-1} L. K. Gallos  and P. Argyrakis, Absence of Kinetic Effects in Reaction-diffusion processes in scale-free Networks, Phys. Rev. Lett., 92 (2004) 138301.


\bibitem{reactive-2} M. Catanzaro, M. Boguna  and R. Pastor-Satorras, Diffusion-annihilation processes in complex networks, Phys. Rev. E, 71 (2005) 056104.

\bibitem{reactive-3} V. Colizza, R. Pastor-Satorras and A. Vespignani,
Reaction-diffusion processes and metapopulation models in heterogeneous networks, Nature Physics, 3 (2007) 276.



\bibitem{jure-networkdata} http://snap.stanford.edu/data/index.html













\bibitem{rumor-model} D. Trpevski, W. K-S. Tang and Lj. Kocarev, A Model for Rumor Spreading over Networks, Physical Review E 81, 056102 (2010).

\bibitem{GoffmanNewill} W. Goffman, V. A. Newill, Generalization of Epidemic Theory: An Application to the Transmission of Ideas, Nature 204, 225-228 (1964)

\bibitem{infoblogspace} D. Gruhl, R. Guha, D. Liben-Nowell and A. Tomkins, Information Diffusion Through Blogspace, Proceedings of the 13th international conference on World Wide Web, ACM New York, NY, USA (2004) 

\end{thebibliography}
\end{document}